\documentclass[10pt]{article}

\usepackage{amssymb,amsmath,amsfonts}
\usepackage{graphicx,color,enumitem}
\usepackage{epsfig}
\usepackage{amsthm} 
\usepackage{bm,dsfont}
\usepackage{subcaption}
\usepackage[round]{natbib}

\usepackage[a4paper, twoside, hmarginratio=1:1, vmarginratio=1:1, left=1in,top=1in]{geometry}

\RequirePackage[breaklinks=true, hidelinks]{hyperref}
\usepackage{breakcites}

 \newcommand{\p}{\mathbb{P}}
 \newcommand{\F}{\mathbb{F}}
 \newcommand{\E}{\mathbb{E}}
 \newcommand{\tildeE}{\mathbb{\tilde E}}
 \newcommand{\R}{\mathbb{R}}
 \newcommand{\cF}{{\mathcal F}}
 \newcommand{\cA}{{\mathcal A}}
 \newcommand{\cU}{{\mathcal U}}

 \newcommand{\id}{\mathds{1}}
 \newcommand{\as}{\mbox{{\rm a.s.}}}

\newtheorem{theorem}{Theorem}[section]

\newtheorem{assumption}[theorem]{Assumption}

\newtheorem{definition}[theorem]{Definition}

\begin{document}
\title{Merton's portfolio problem under Volterra Heston model}
\author{Bingyan Han \thanks{Department of Statistics, The Chinese University of Hong Kong, Hong Kong, byhan@link.cuhk.edu.hk}
	\and Hoi Ying Wong\thanks{Department of Statistics, The Chinese University of Hong Kong, Hong Kong, hywong@cuhk.edu.hk}
}
\date{November 18, 2019}

\maketitle

\begin{abstract}
This paper investigates Merton's portfolio problem in a rough stochastic environment described by Volterra Heston model. The model has a non-Markovian and non-semimartingale structure. By considering an auxiliary random process, we solve the portfolio optimization problem with the martingale optimality principle. Optimal strategies for power and exponential utilities are derived in semi-closed form solutions depending on the respective Riccati-Volterra equations. We numerically examine the relationship between investment demand and volatility roughness.
	\\[2ex] 
	\\[2ex]
\end{abstract}



\section{Introduction}
Empirical studies suggest that implied and realized volatilities of major financial indices tend to have rougher sample paths than the ones modeled by the standard Brownian motion \citep{gatheral2018volatility}. This discovery stimulates a rapidly growing development in rough volatility models recently. This new generation of models are constructed with stochastic processes such as fractional Brownian motion (fBm), fractional Ornstein-Uhlenbeck (fOU) process, and rough Bergomi (rBergomi) model. Recent empirical studies suggest that roughness is also associated with market fear \citep{caporale2018fear} and risks in pension fund portfolios \citep{cadoni2017pension}.

The popularity of the Heston model in the financial market leads to the introduction of the fractional Heston model  \citep{guennoun2018asym} and the rough Heston model \citep{eleuch2019char}. Both are rough versions of the celebrated Heston stochastic volatility model. Compared with classic Heston model, rough Heston model \citep{eleuch2019char} builds on market microstructure and better captures the explosion of at-the-money (ATM) skew, as illustrated in Figures (\ref{Fig:ATM1})-(\ref{Fig:ATM2}). Other motivations and properties of rough volatility models can be found in \cite{gatheral2018volatility,eleuch2019char}. Recent remarkable advances include the derivation of the characteristic function of the rough Heston model \citep{eleuch2019char} and the affine Volterra processes \citep{abi2017affine}. The Volterra Heston model serves as an important specific example in \cite{abi2017affine}. In addition, the rough Heston model becomes a special case of Volterra Heston model under the fractional kernel. The structure of characteristic functions in \cite{eleuch2019char} can be extended to affine Volterra processes using Riccati-Volterra equations as shown in \cite{abi2017affine}. Therefore, this paper focuses on the financial market with the Volterra Heston model.

\begin{figure}[!h]
	\centering
	\begin{minipage}{0.5\textwidth}
		\centering
		\includegraphics[width=0.9\textwidth]{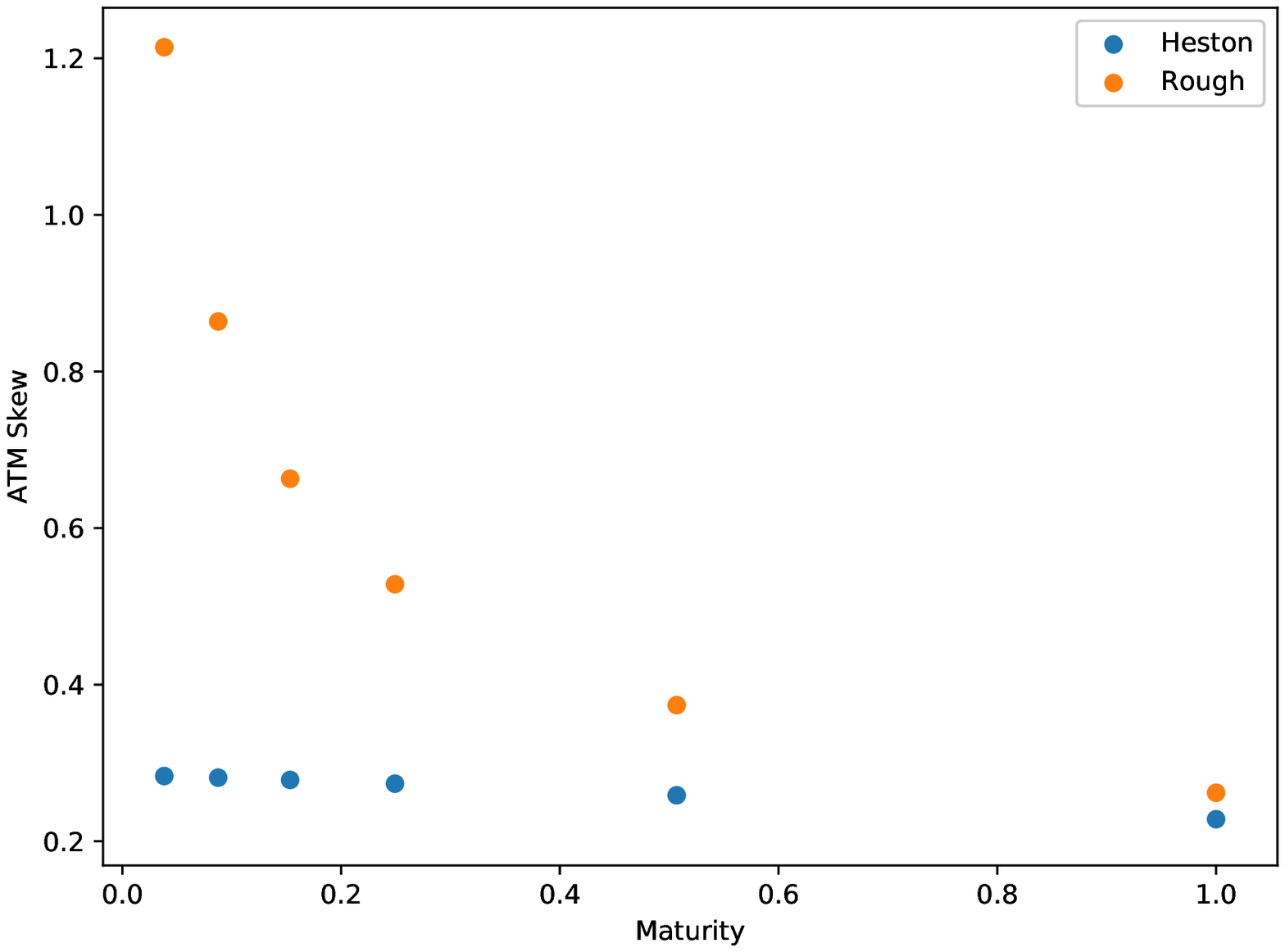}
		\subcaption{ATM skew of Heston and rough Heston model}\label{Fig:ATM1}
	\end{minipage}%
	\begin{minipage}{0.5\textwidth}
		\centering
		\includegraphics[width=0.9\textwidth]{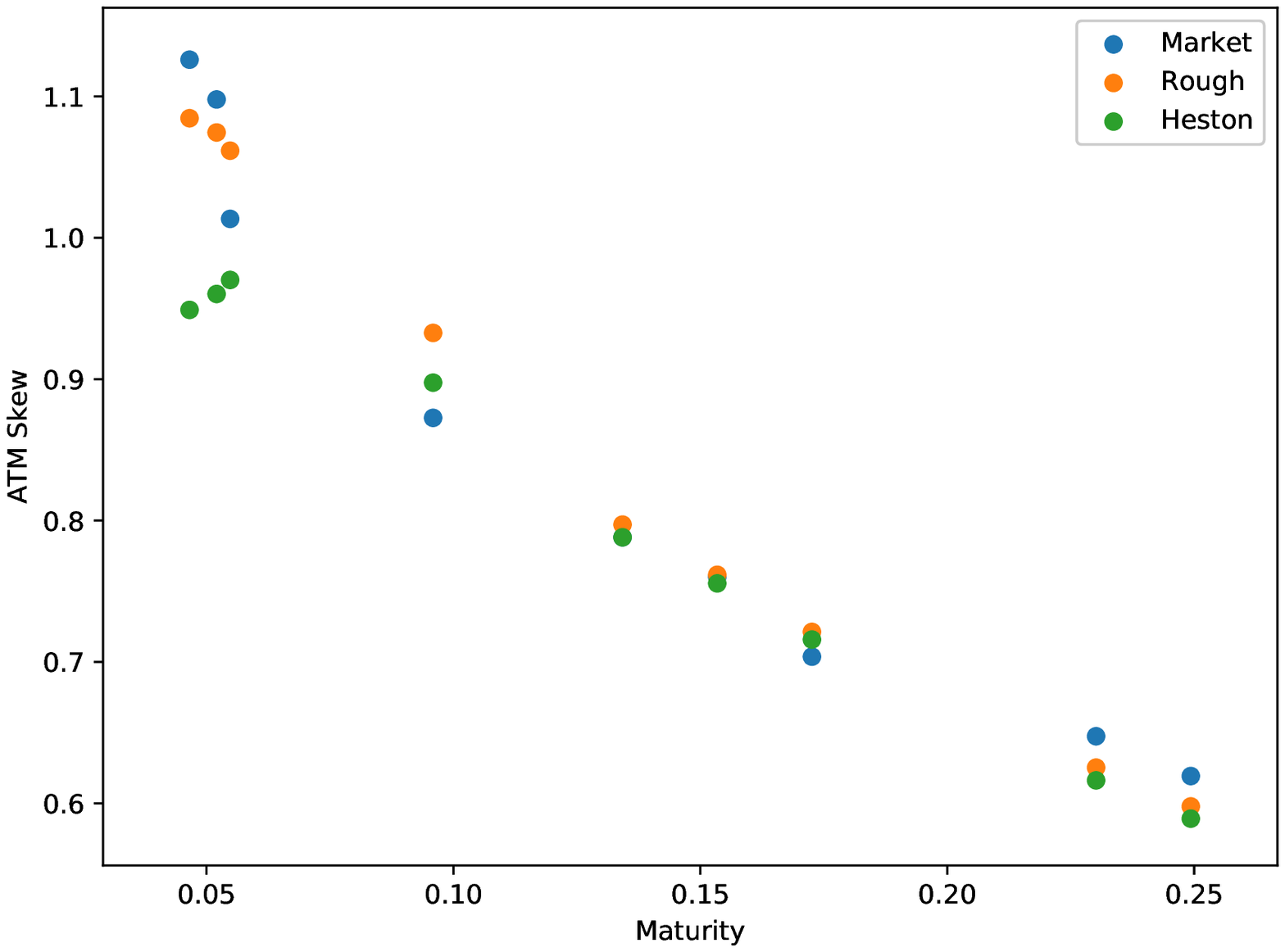}
		\subcaption{Calibration}\label{Fig:ATM2}
	\end{minipage}
	\caption{ATM skew. Left subplot gives a theoretical comparison. For rough Heston model, parameters are $H = 0.12$,  $\phi = 0.3156$,  $\kappa = 0.1$, $ \rho = -0.681$,  $\sigma = 1.044$, $V_0 = 0.0392$. For classic Heston model, let $H= 0.5$ and others unchanged. Rough Heston model captures the explosion of ATM skew easily. The right subplot illustrates calibration based on S\&P 500 implied volatility surface on 03/29/2019.  For rough Heston model, parameters are $H = 0.118$,  $\phi = 0.170$,  $\kappa = 0.173$, $ \rho = -0.615$,  $\sigma = 0.340$, $V_0 = 0.0182$. For classic Heston model, $\phi = 0.037$,  $\kappa = 5.609$, $\rho = - 0.614$,  $\sigma = 1.188$,  $V_0 = 0.015$. Classic Heston model obtains larger $\kappa$ and $\sigma$, but it still fails to model the explosion of ATM skew, observed from the market data when the maturity approaches zero. We refer readers to (\ref{vol}) for notations.}
\end{figure}

How does the roughness of the market volatility affect investment demands? We address this by investigating the optimal investment demand with the Merton problem as it is probably the most classic financial economic approach to do so \citep{merton1969lifetime,kraft2005optimal,hu2005utility,jang2014ambiguity}. The literature tends to focus more on the option pricing problems and portfolio optimization under rough volatility models is still at an early stage. However, some recent works do exist, see \cite{fouque2018aoptimal, fouque2018boptimal, bauerle2018portfolio, glasserman2018buy, han2019mean,han2019time} and references therein. The studies in \cite{fouque2018aoptimal, fouque2018boptimal} consider the expected power utility portfolio maximization with slow or fast varying stochastic factors driven by the fOU processes whereas the fractional Heston model \citep{guennoun2018asym} is used in \cite{bauerle2018portfolio} with the same objective function. Inspired by some market insights, it is suggested in \cite{glasserman2018buy} to use roughness as a trading signal. To the best of our knowledge, portfolio selection with Volterra Heston model is firstly studied in \cite{han2019mean} in the context of mean-variance objective.

In this paper, we investigate Merton's portfolio problem with an unbounded risk premium. To overcome the difficulty from the non-Markovian and non-semimartingale characteristic in Volterra Heston model, we apply the martingale optimality principle \citep{hu2005utility,pham2009book,jeanblanc2012mean} and construct the Ansatz, which is inspired by the martingale distortion transformation \citep{zariphopoulou2001solution, fouque2018aoptimal} and the exponential-affine representations \citep{abi2017affine}. The key finding is the auxiliary processes $M_t$ in (\ref{Eq:M}) and (\ref{Eq:ExpM}) with properties presented in Theorem \ref{Thm:M} and \ref{Thm:ExpM} below. We offer explicit solutions to the optimal portfolio policies that depend on Riccati-Volterra equations, which can be solved by well-known numerical methods. Our result differs from previous literature as follows:
\begin{itemize}
	\item \cite{hu2005utility,fouque2018aoptimal,fouque2018boptimal} assume an essentially bounded risk premium. This assumption does not hold for the Volterra Heston model so that their results cannot be applied here.
	\item Compared with \cite{bauerle2018portfolio}, we consider different  models and provide explicit solutions for a non-zero correlation between stock and volatility, reflecting the well-known market leverage effect. 
	\item Although the power utility maximization with the classic Heston model has been studied in \cite{kraft2005optimal}, the non-Markovian and non-semimartingale characteristic in Volterra Heston model prevents the use of the Hamilton-Jacobi-Bellman (HJB) framework for the classic model in \cite{kraft2005optimal}.
	\item The $M_t$ in (\ref{Eq:M}) is unbounded, making the proof of Theorem \ref{Thm:M} and \ref{Thm:alpha*} different from both \cite{fouque2018aoptimal} and \cite{han2019mean}.
	\item We also consider the exponential utility. The corresponding wealth process is not necessarily bounded below. It complicates the proof in Theorem \ref{Thm:Expu}. Moreover, we provide an interesting comparison between power utility and exponential utility, from the perspective of rough volatility.
	\item \cite{han2019mean} apply completion of squares technique which is tailor-made for mean-variance portfolio (MVP) problem, while this paper adopts martingale optimality principle. Mean-variance portfolio selection and expected utility theory are usually considered separately, due to the completely different objectives and solution methods. In addition, dynamic MVP may have to additionally address the issue of time inconsistency \citep{chen2019time,han2019time}.
\end{itemize}

By deriving the analytical optimal investment demand, we are able to discuss the effect of volatility roughness on investment decisions. In rough Heston market, our sensitivity analysis suggests that investors with a power utility demand less on the stock market if the stock volatility is rougher whereas investors with an exponential utility demand more although both types of investors are risk averse. We further expand our analysis to parameters calibrated to simulated option data. The behaviors remain different for these two types of investors, documenting a significant difference in economic results for these two utility functions. 

The rest of this paper is organized as follows. We present the problem formulation in Section \ref{Sec:Pro} and solve the problem by the martingale optimality principle in Section \ref{Sec:Sol}. Section \ref{Sec:Num} offers numerical illustration for the investment demand under the rough Heston model. Section \ref{Sec:Con} concludes. All technical proofs are put into Appendix \ref{Appednix}.

\section{Problem formulation}\label{Sec:Pro}
Let $(\Omega, \cF, \p)$ be a complete probability space, with a filtration $\F = \{ \cF_t \}_{ 0 \leq t \leq T}$ satisfying the usual conditions, supporting a two-dimensional Brownian motion $W = (W_1, W_2)$. The filtration $\F$ is not necessarily the augmented filtration generated by $W$.

Denote a kernel $K(\cdot) \in L^2_{loc} (\R_+, \R)$ where $\R_+ = \{ t \in \R | t \geq 0\}$. Suppose the standing Assumption \ref{Assum:K} holds throughout the paper, in line with \cite{abi2017affine,keller2018affine,han2019mean,han2019time}. Recall that a function $f$ is completely monotone on $(0, \infty)$, if it is infinitely differentiable on $(0, \infty)$ and $(-1)^k f^{(k)}(t) \geq 0$ for all $ t > 0 $, $k = 0, 1, ...$. Assumption \ref{Assum:K} is satisfied by positive constant kernels, fractional kernels and exponential kernels, with proper parameters, see \cite{abi2017affine}.

\begin{assumption}\label{Assum:K}
The kernel $K$ is strictly positive and completely monotone on $(0, \infty)$. There is $\delta \in(0,2]$ such that $\int_{0}^{h} K(t)^{2} d t=O\left(h^\delta \right)$ and $\int_{0}^{T}(K(t+h)-K(t))^{2} d t=O(h^\delta)$ for every $T<\infty$. 
\end{assumption}

The convolution $K*L$ for a measurable kernel $K$ on $\R_+$ and a measure $L$ on $\R_+$ of locally bounded variation is defined by
\begin{equation}
(K*L)(t) = \int^t_0 K(t-s)L(ds)
\end{equation}
for $t>0$ under proper conditions. The integral is extended to $t=0$ by right-continuity if possible. If $F$ is a function on $\R_+$, let
\begin{equation}
(K*F)(t) = \int_0^t K(t-s) F(s) ds.
\end{equation}

For a $1$-dimensional continuous local martingale $W$, the convolution between $K$ and $W$ is defined as
\begin{equation}
(K*dW)_t = \int_0^t K(t-s)dW_s.
\end{equation}

A measure $L$ on $\R_+$ is called {\em resolvent of the first kind} to $K$, if
\begin{equation}
K*L = L*K \equiv \rm{id}.
\end{equation}
The existence of a resolvent of the first kind is shown in \citet[Theorem 5.5.4]{gripenberg1990volterra} under the complete monotonicity assumption, imposed in Assumption \ref{Assum:K}. Alternative conditions for the existence are given in \citet[Theorem 5.5.5]{gripenberg1990volterra}.

Kernel $R$ is called the {\em resolvent}, or {\em resolvent of the second kind}, to $K$ if
\begin{equation}
K*R = R*K = K - R.
\end{equation}
The resolvent always exists and is unique by \citet[Theorem 2.3.1]{gripenberg1990volterra}. Further properties of these definitions can be found in \cite{gripenberg1990volterra, abi2017affine}. Examples of kernels are available at \citet[Table 1]{abi2017affine}.

The variance process of the Volterra Heston model is defined as
\begin{equation}\label{vol}
V_{t}=V_{0}+ \kappa \int_{0}^{t} K(t-s)\left(\phi -V_{s}\right) d s + \int_{0}^{t} K(t-s) \sigma \sqrt{V_{s}} d B_{s},
\end{equation}
where $ dB_s = \rho dW_{1s} + \sqrt{1 - \rho^2} dW_{2s} $ and $V_0, \kappa, \phi, \sigma$ are positive constants. The correlation $\rho$ between stock price and variance is also assumed constant. The process in (\ref{vol}) is non-Markovian and non-semimartingale in general. Rough Heston model in \cite{eleuch2019char, eleuch2018perfect} is a special case of (\ref{vol}) with $K(t) = \frac{t^{H - 1/2}}{\Gamma(H + 1/2)}$, $H \in (0, 1/2]$. An alternative definition for Heston model with rough paths in \cite{guennoun2018asym} is known as fractional Heston model.

Suppose there is a risk-free asset with deterministic bounded risk-free rate $r_t >0$. Following  \cite{abi2017affine, kraft2005optimal}, we assume the risky asset (stock or index) price $S_t$ follows
\begin{equation}\label{stock}
dS_t = S_t (r_t + \theta V_t) dt + S_t \sqrt{V_t} dW_{1t}, \quad S_0 > 0,
\end{equation}
with constant $\theta \neq 0$. Then the market price of risk (risk premium) is given by $\theta \sqrt{V_t}$. 

We need the following existence and uniqueness result.
\begin{theorem}
	(Theorem 7.1 in \cite{abi2017affine}) Under Assumption \ref{Assum:K}, the stochastic Volterra equation (\ref{vol})-(\ref{stock}) has a unique in law $\R_+ \times \R_+$-valued continuous weak solution for any initial condition $(S_0, V_0) \in \R_+ \times \R_+$.
\end{theorem}
 For strong uniqueness, we mention \citet[Proposition B.3]{abi2019multifactor} as a related result with kernel $K \in C^1([0, T], \R)$ and \citet[Proposition 8.1]{mytnik2015uniqueness} for certain Volterra integral equations with smooth kernels. The strong uniqueness of (\ref{vol})-(\ref{stock}) is left open for singular kernels. For weak solutions, Brownian motion is also a part of the solution. However, expected utility only depends on the expectation of the wealth process. In the sequel, we fix a version of the solution $(S, V, W_1, W_2)$ to (\ref{vol})-(\ref{stock}) as other solutions have the same law.

Merton problem aims at
\begin{equation}\label{Genobj}
\sup_{ \alpha(\cdot) \in \cA } \E\left[ U(X_T) \right],
\end{equation}
where $\alpha$ is an admissible investment strategy introduced later and $X$ is the corresponding wealth. $U(\cdot)$ is a utility function. Specifically, we consider power utility and exponential utility in this paper. The classic martingale optimality principle, see, e.g., \cite{hu2005utility}, \citet[Section 6.6.1]{pham2009book} or \cite{jeanblanc2012mean}, states that the Problem (\ref{Genobj}) can be solved by constructing a family of processes $\{ J^\alpha_t \}_{ t \in [0, T]}$, $\alpha \in \cA$, satisfying conditions:
\begin{enumerate}[label={(\arabic*).}]
	\item $J^\alpha_T = U(X_T)$ for all $\alpha \in \cA$; 
	\item $J^\alpha_0$ is a constant, independent of  $\alpha \in \cA$;
	\item $J^\alpha_t$ is a supermartingale for all $\alpha \in \cA$, and there exists  $\alpha^* \in \cA$ such that $J^{\alpha^*}$ is a martingale.
\end{enumerate}
Indeed, if we can find $J^\alpha_t$, then for all $\alpha \in \cA$,
\begin{equation*}
\E[ U(X_T) ] = \E[ J^\alpha_T ] \leq J^\alpha_0 = J^{\alpha^*}_0 = \E[ J^{\alpha^*}_T]  = \E[ U(X^*_T)],
\end{equation*}
where $X^*$ is the wealth process under $\alpha^*$.

\section{Optimal strategy}\label{Sec:Sol}
\subsection{Power utility}
Let $\alpha_t \triangleq \sqrt{V_t} \pi_t$ be the investment strategy, where $\pi_t$ is the proportion of wealth invested in the stock. Then, the wealth process $X_t$ reads
\begin{equation}\label{Eq:wealth}
d X_t = \big(r_t + \theta \sqrt{V_t} \alpha_t \big)  X_t dt + \alpha_t X_t dW_{1t}, \; X_0 = x_0 > 0.
\end{equation}

\begin{definition}
	An investment strategy $\alpha(\cdot)$ is said to be admissible if 
	\begin{enumerate}[label={(\arabic*).}]
		\item $\alpha(\cdot)$ is $\F$-adapted and $\int^T_0 \alpha^2_s ds < \infty$, $\p$-$\as$;
		\item the wealth process (\ref{Eq:wealth}) has a unique solution in the sense of \citet[Chapter 1, Definition 6.15]{yong1999book}, with $\p$-$\as$ continuous paths;
		\item $X_t \geq 0$,  $\forall \; t \in [0, T]$, $\p$-$\as$;
		\item $\E\Big[ \frac{1}{\gamma} X^\gamma_T \Big] < \infty$,  $0 < \gamma < 1$.
	\end{enumerate}
	The set of all admissible investment strategies is denoted as $\cA$.
\end{definition}

We are interested in the power utility optimization:
\begin{equation}\label{obj}
\sup_{ \alpha(\cdot) \in \cA } \E\left[ \frac{1}{\gamma} X^\gamma_T \right], \quad 0 < \gamma < 1.
\end{equation}
To ease notation burden, we simply write $X$, instead of $X^{x_0, \alpha}$, as the wealth process (\ref{Eq:wealth}) under $\alpha \in \cA$ with initial condition $X_0 = x_0 > 0$. 
 
To construct $J^\alpha_t$, we introduce a new probability measure $ \tilde \p$ together with $\tilde W_{1t} \triangleq W_{1t}  - \frac{\gamma \theta}{1 - \gamma}  \int^t_0 \sqrt{V_s} ds$ as the new Brownian motion by \citet[Lemma 7.1]{abi2017affine}. Under $\tilde \p$, 
\begin{equation}
V_{t}=V_{0}+ \int_{0}^{t} K(t-s)\left( \kappa \phi - \lambda V_{s}\right) d s + \int_{0}^{t} K(t-s) \sigma \sqrt{V_{s}} d \tilde B_{s},
\end{equation}
with $\lambda = \kappa - \frac{\gamma}{1 - \gamma} \rho  \theta \sigma$ and $d \tilde B_s = \rho d\tilde W_{1s} + \sqrt{1- \rho^2} d W_{2s}$.

Denote $\tilde \E[\cdot]$ and $\tilde \E_t[ \cdot ] = \tildeE[ \cdot | \cF_t]$ as the $\tilde \p$-expectation and conditional $\tilde \p$-expectation,  respectively. The forward variance under $\tilde \p$ is the conditional $\tilde \p$-expected variance, that is, $ \tildeE_t [ V_s ] \triangleq \xi_{t}(s)$. It is shown in \citet[Propsition 3.2]{keller2018affine} and \citet[Lemma 4.2]{abi2017affine} that 
\begin{equation}\label{Eq:xi}
\xi_{t}(s) = \tildeE \left[V_{s} | \mathcal{F}_{t}\right]=\xi_{0}(s)+\int_0^t \frac{1}{\lambda} R_{\lambda}(s-u) \sigma \sqrt{V_{u}} d \tilde B_{u},
\end{equation}
where
\begin{equation}
\xi_{0}(s) = \left(1-\int_{0}^{s} R_{\lambda}(u) d u\right) V_{0} + \frac{\kappa \phi}{\lambda} \int_{0}^{s} R_{\lambda}(u) du,
\end{equation}
and $R_\lambda$ is the resolvent of $\lambda K$ such that 
\begin{equation}
\lambda K * R_\lambda = R_\lambda * ( \lambda K) = \lambda K - R_\lambda.
\end{equation}

Consider the stochastic process, 
\begin{equation}\label{Eq:M}
M_t = \exp \Big[ \int^T_t \big(\gamma r_s + \frac{\gamma \theta^2 \xi_t(s)}{2(1-\gamma)}  + \frac{c\sigma^2}{2} \psi^2(T-s) \xi_t(s) \big) ds \Big],
\end{equation}
where $ c =\frac{1 - \gamma}{ 1 -\gamma + \gamma \rho^2} $ and $\psi(\cdot)$ satisfies the Riccati-Volterra equation
\begin{align}\label{Eq:psi}
\psi(t) & = \int^t_0 K(t - s) \big[ \frac{\sigma^2}{2} \psi^2(s)  - \lambda \psi(s) + \frac{\gamma \theta^2}{2c(1-\gamma)} \big] ds. 
\end{align}
Existence and uniqueness of the solution to (\ref{Eq:psi}) are established in \citet[Lemma A.2 and A.3]{han2019mean} based on the results of \cite{gatheral2018affine, eleuch2018perfect}. Indeed, if $\lambda > 0$ and $\lambda^2 - \frac{\gamma \theta^2 \sigma^2}{c (1-\gamma)} > 0$, then (\ref{Eq:psi}) has a unique non-negative global solution. These assumptions are also in line with \citet[Proposition 5.2]{kraft2005optimal}. Furthermore, there is a tighter result for (\ref{Eq:psi}) with the fractional kernel in \citet[Theorem 3.2]{eleuch2018perfect}.

By considering $M_t$, we overcome the non-Markovian and non-semimartingale difficulty in the variance process (\ref{vol}). Main properties of $M_t$ are summarized in Theorem \ref{Thm:M}. We highlight that the $M_t$ in (\ref{Eq:M}) is unbounded so that it is very different from the one considered in \cite{fouque2018aoptimal, han2019mean}.
\begin{theorem}\label{Thm:M}
	Assume 
	\begin{equation}\label{Assum:M}
	\kappa^2 - 6 \frac{\gamma^2}{(1 - \gamma)^2} \theta^2 \sigma^2 > 0, \quad \lambda > 0, \quad \lambda^2 - 2 p \frac{ \gamma}{1 - \gamma} \theta^2 \sigma^2 > 0,
	\end{equation}
	for some $ p > 1/(2c)$. Then $M$ has following properties:
	\begin{enumerate}[label={(\arabic*).}]
		\item $M_t \geq l > 0$ for some positive constant $l$. And $\E\big[ \sup_{ t \in [0, T]} |M_t|^p \big] < \infty$;
		\item Apply It\^o's lemma to $M$ on $t$, then
		\begin{align}
		dM_t =& - \big[ \gamma r_t + \frac{\gamma}{2(1-\gamma)}\theta^2 V_t \big] M_t dt - \frac{\gamma}{2(1-\gamma)} \big[ 2 \theta \sqrt{V_t} U_{1t} + \frac{U^2_{1t}}{M_t} \big] dt + U_{1t} dW_{1t} + U_{2t} dW_{2t},
		\end{align}
		where
		\begin{align}
		U_{1t} &= \rho c \sigma M_t \sqrt{V_t} \psi(T-t), \label{Eq:U1simp}\\
		U_{2t} &= \sqrt{1-\rho^2} c \sigma M_t \sqrt{V_t} \psi(T-t). \label{Eq:U2simp}
		\end{align}
		\item $\E\Big[ \big(\int^T_0 U^2_{it} dt \big)^{p/4} \Big] < \infty$ for $ i = 1, 2$.
	\end{enumerate}
\end{theorem}

Now we are ready to give the Ansatz for $J^\alpha_t$. Consider
\begin{equation}
J^\alpha_t = \frac{X^\gamma_t}{\gamma} M_t.
\end{equation}
Then we have the following verification result.
\begin{theorem}\label{Thm:alpha*}
	Suppose the conditions (\ref{Assum:M}) in Theorem \ref{Thm:M} hold and $\kappa^2 - 2 \eta \sigma^2 > 0$, where
	\begin{equation}
	\eta = \max\Big\{ 2q |\theta| \sup_{ t \in [0, T]} |A_t|, (8q^2 - 2q)  \sup_{ t \in [0, T]} |A_t|^2 \Big\},
	\end{equation}
	for some $q > 1$ and $A_t = \frac{1}{1 -\gamma} \big[ \theta + \rho c \sigma \psi(T-t) \big]$. Then $J^\alpha_t = \frac{X^\gamma_t}{\gamma} M_t$ satisfies the martingale optimality principle, and an optimal strategy is given by 
	\begin{equation}\label{Eq:alpha*}
	\alpha^*_t = A_t \sqrt{V_t}.
	\end{equation}
	 Moreover,
	\begin{equation}
	\E\Big[ \sup_{ t \in [0, T]} |X^*_t|^q \Big] < \infty,
	\end{equation}
	and $\alpha^*$ is admissible.
\end{theorem}

\subsection{Exponential utility}
In this subsection, we consider the exponential utility case. With a slightly different formulation, let $u_t$ be the investment strategy. The wealth process $X_t$ reads
\begin{equation}\label{Eq:Expwealth}
d X_t = \big(r_t X_t + \theta \sqrt{V_t} u_t \big) dt + u_t dW_{1t}, \; X_0 = x_0 > 0.
\end{equation}
The following admissibility assumption is consistent with \cite{hu2005utility}.
\begin{definition}
	Let $\gamma > 0$, an investment strategy $u(\cdot)$ is said to be admissible if 
	\begin{enumerate}[label={(\arabic*).}]
		\item $u(\cdot)$ is $\F$-adapted and $ \E[\int^T_0 u^2_s ds] < \infty$;
		\item the wealth process (\ref{Eq:Expwealth}) has a unique solution in the sense of \citet[Chapter 1, Definition 6.15]{yong1999book}, with $\p$-$\as$ continuous paths;
		\item $\{\exp\big[ - \gamma e^{ \int^T_\tau r_v dv} X_\tau \big]: \tau \text{ stopping time with values in } [0, T] \}$ is a uniformly integrable family.
	\end{enumerate}
	The set of all admissible investment strategies is denoted as $\cU$.
\end{definition}
The investor now considers
\begin{equation}\label{Expobj}
\sup_{u(\cdot) \in \cU } \E\left[ - \frac{1}{\gamma} e^{- \gamma X_T} \right], \quad \gamma > 0.
\end{equation}
The solution method is similar to the power utility case. With slightly abuse of notations, we still use $\psi(\cdot)$, forward variance $\xi_t(s)$, $M$, etc. However, they are redefined and not mixed with counterparts in power utility case.

To construct $J^u_t$, we introduce a new probability measure $ \tilde \p$ together with $\tilde W_{1t} \triangleq W_{1t}  + \theta  \int^t_0 \sqrt{V_s} ds$ as the new Brownian motion and $\lambda = \kappa + \rho \sigma \theta$. Consider the stochastic process, 
\begin{equation}\label{Eq:ExpM}
M_t = \exp \Big[ \int^T_t \big(- \frac{\theta^2}{2} \xi_t(s) + \frac{\sigma^2(1 - \rho^2)}{2} \psi^2(T-s) \xi_t(s) \big) ds \Big],
\end{equation}
where  $\xi_t(s)$ is forward variance under $ \tilde \p$ and $\psi(\cdot)$ satisfies
\begin{align}\label{Eq:Exppsi}
\psi(t) & = \int^t_0 K(t - s) \big[ \frac{\sigma^2(1 - \rho^2)}{2} \psi^2(s)  - \lambda \psi(s) - \frac{\theta^2}{2} \big] ds. 
\end{align}
If $\lambda > 0$, then (\ref{Eq:Exppsi}) has a unique global solution.
\begin{theorem}\label{Thm:ExpM}
	Assume $\lambda > 0$. Then $M$ in (\ref{Eq:ExpM}) has following properties:
	\begin{enumerate}[label={(\arabic*).}]
		\item $M_t$ is essentially bounded. $0 < M_t \leq 1$, $\forall \; t \in [0, T]$, $\p$-$\as$;
		\item Apply It\^o's lemma to $M$ on $t$, then
		\begin{equation}\label{Eq:ItoExpM}
		dM_t = \frac{(\theta \sqrt{V_t} M_t + U_{1t})^2}{2 M_t} dt + U_{1t} dW_{1t} + U_{2t} dW_{2t},
		\end{equation}
		where
		\begin{equation}
		U_{1t} = \rho \sigma M_t \sqrt{V_t} \psi(T-t), \quad U_{2t} = \sqrt{1-\rho^2} \sigma M_t \sqrt{V_t} \psi(T-t).
		\end{equation}
		\item $\E\Big[ \big(\int^T_0 U^2_{it} dt \big)^{p/2} \Big] < \infty$ for $p \geq 1$ and $ i = 1, 2$.
	\end{enumerate}
\end{theorem}

Consider
\begin{equation}
J^u_t = - \frac{1}{\gamma} \exp \big[ - \gamma e^{\int^T_t r_sds} X_t\big] M_t.
\end{equation}
Then we have the following verification result.
\begin{theorem}\label{Thm:Expu}
	Suppose $\lambda > 0$ and $\kappa^2 - 2 \eta \sigma^2 > 0$ with
	\begin{equation}
	\eta =  \sup_{t \in [0, T]} \Big\{ 2 p^2 \gamma^2 e^{ \int^T_t 2 r_v dv} A^2_t + 2 p \gamma e^{ \int^T_t r_v dv} |\theta A_t| \Big\}, \text{ for some } p > 1,
	\end{equation}
	 where $A_t = \frac{1}{\gamma} e^{- \int^T_t r_v dv} \big[ \theta + \rho \sigma \psi(T-t) \big]$. Then $J^u_t$ satisfies the martingale optimality principle. An optimal strategy is given by 
	\begin{equation}\label{Eq:Expu}
	u^*_t = A_t \sqrt{V_t},
	\end{equation}
	and $u^*$ is admissible.
\end{theorem}

\section{Investment under rough volatility}\label{Sec:Num}
\begin{figure}[!h]
	\centering
	\begin{minipage}{0.5\textwidth}
		\centering
		\includegraphics[width=0.9\textwidth]{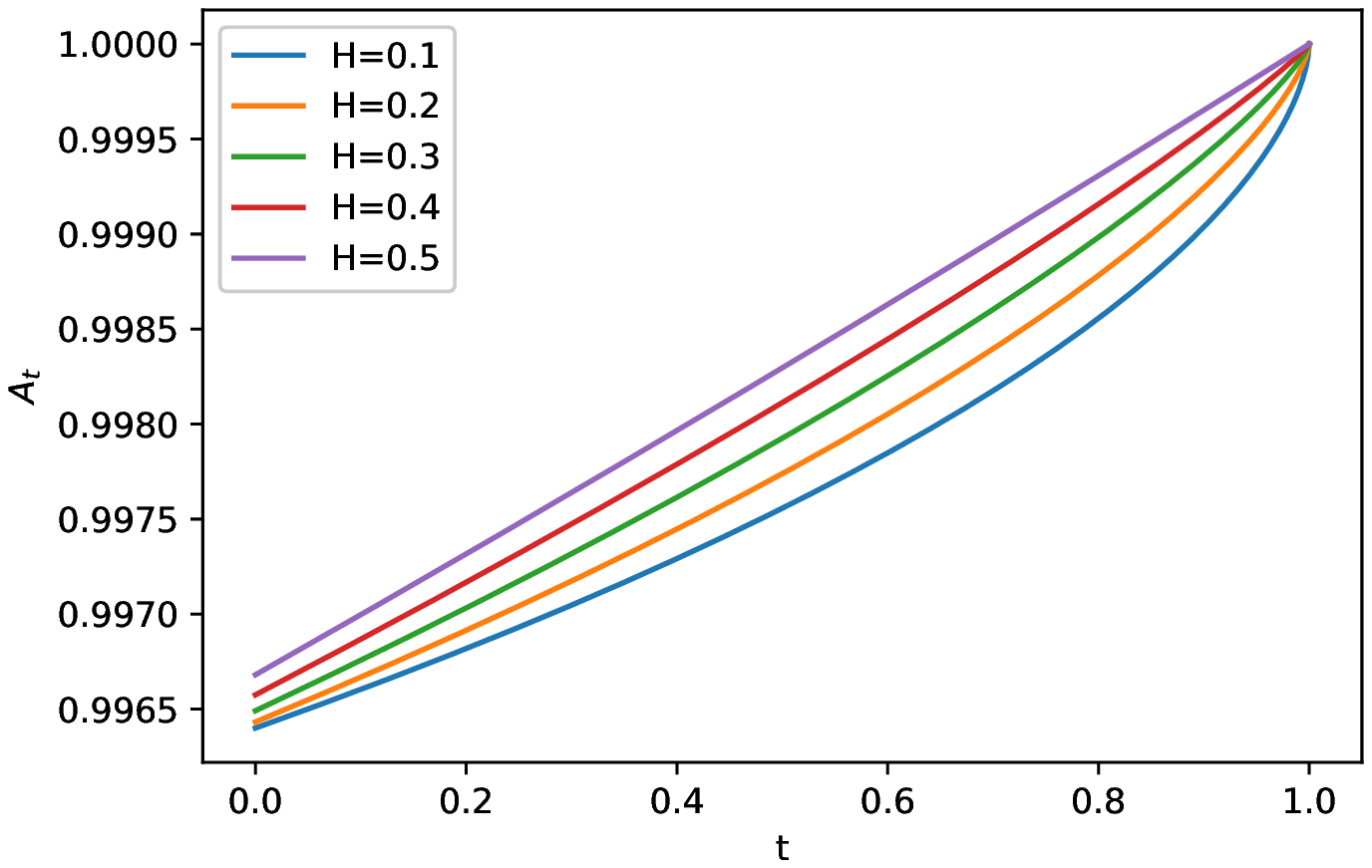}
		\subcaption{$A_t$ under power utility}\label{Fig:PowA}
	\end{minipage}%
	\begin{minipage}{0.5\textwidth}
		\centering
		\includegraphics[width=0.9\textwidth]{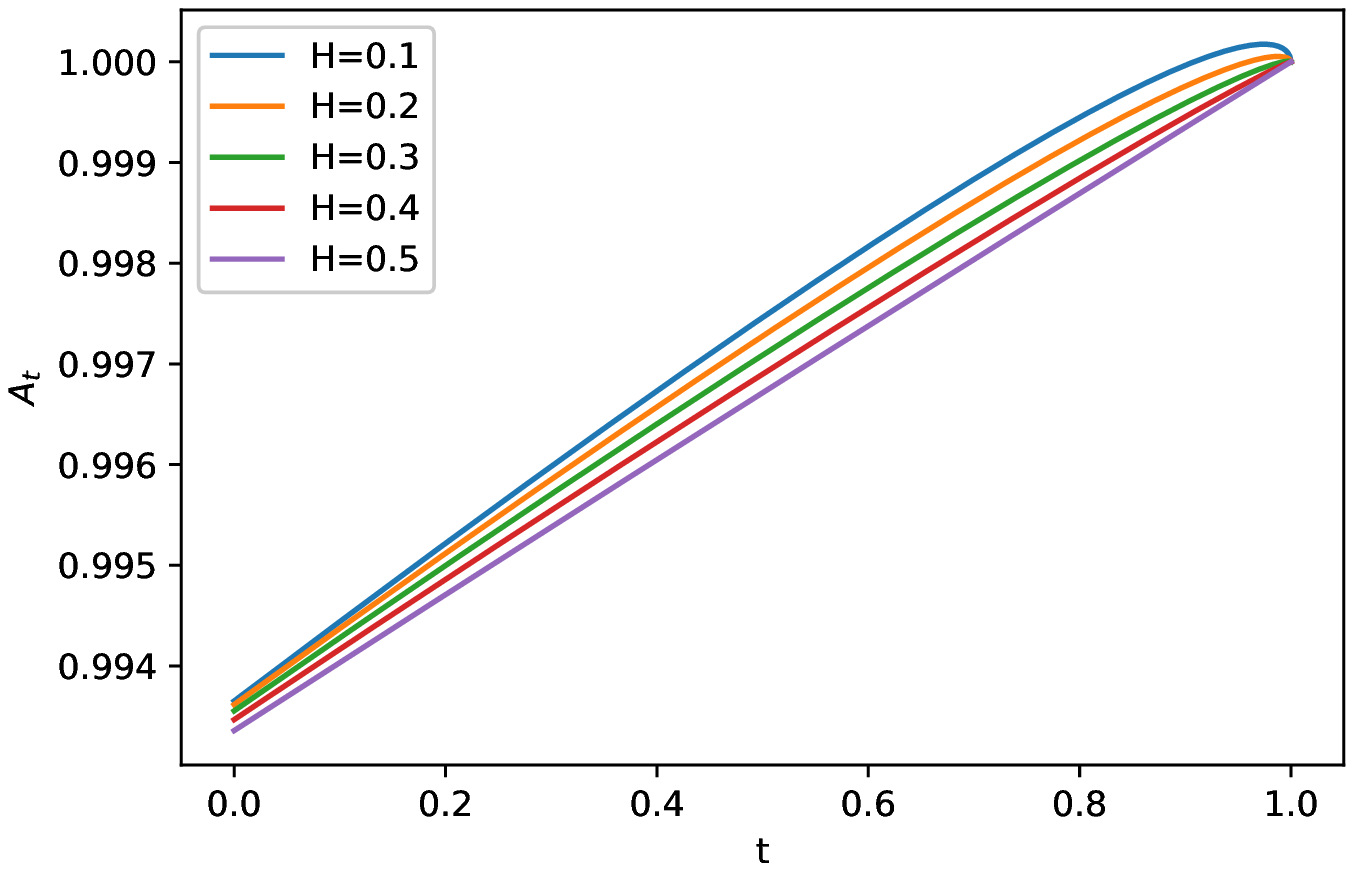}
		\subcaption{$A_t$ under exponential utility}\label{Fig:ExpA}
	\end{minipage}
	\caption{$A_t$ under different utility and $H$. We set risk-free rate $r=0.01$, risk aversion $\gamma = 0.5$, volatility of volatility $\sigma = 0.02$, mean-reversion speed $\kappa = 0.1$, risk premium parameter $\theta = 0.5$, correlation $\rho = -0.7$, and time horizon $T = 1$.}
\end{figure}

\begin{figure}[!h]
	\centering
	\begin{minipage}{0.5\textwidth}
		\centering
		\includegraphics[width=0.9\textwidth]{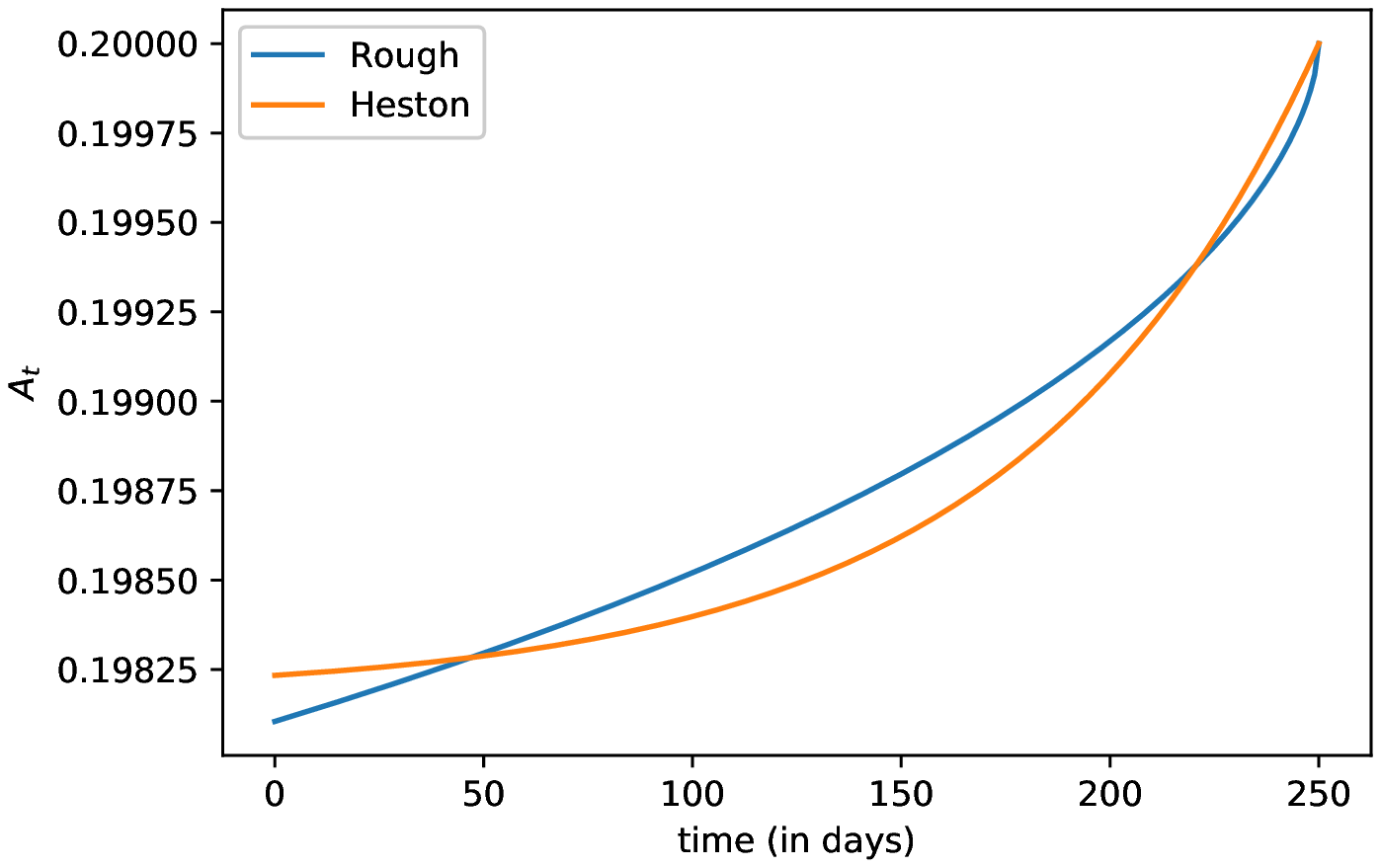}
		\subcaption{$A_t$ under power utility}\label{Fig:PowIV}
	\end{minipage}%
	\begin{minipage}{0.5\textwidth}
		\centering
		\includegraphics[width=0.9\textwidth]{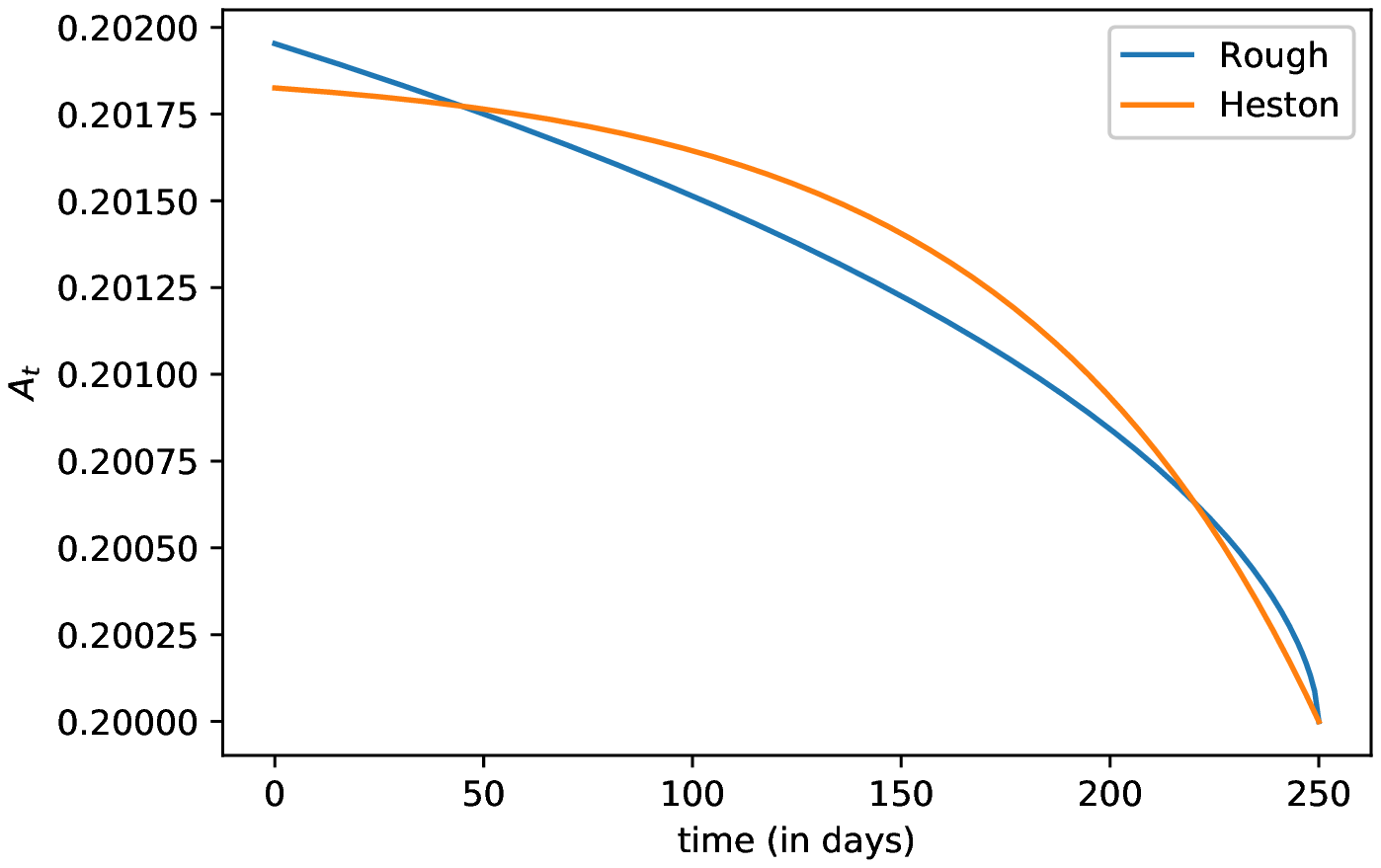}
		\subcaption{$A_t$ under exponential utility}\label{Fig:ExpIV}
	\end{minipage}
	\caption{We set $250$ time steps for one year, corresponding to the $250$ trading days in a year. Risk-free rate $r=0$. Risk premium parameter $\theta=0.1$. Risk aversion $\gamma = 0.5$. Other parameters are adopted from \citet[Table 6]{abi2019lifting} and \citet[Table 4]{abi2019lifting}. }
\end{figure}

Both power utility and exponential utility have been widely adopted in literature. Different behaviors of these utility functions have been demonstrated before. Power utility has constant relative risk aversion, while exponential utility has constant absolute risk aversion. Optimal strategy under power utility is proportional to wealth. Therefore, it guarantees a non-negative wealth and avoids bankruptcy. In contrast, optimal strategy under exponential utility is not related to investor's wealth. The proportion of wealth in the stock is lower for richer investors, since the dollar value in the stock is the same.  \cite{merton1969lifetime} wrote that exponential utility is ``behaviorially less plausible than constant relative risk aversion". In our numerical study, we discover a new difference of these two utility functions from the perspective of roughness.

We consider the fractional kernel $K(t) = \frac{t^{H - 1/2}}{\Gamma(H + 1/2)}$, $H \in (0, 1/2)$, corresponding to rough Heston model. The smaller the Hurst parameter $H$, the rougher the volatility of the stock. $H = \frac{1}{2}$ corresponds to the strategy in \cite{kraft2005optimal} under the classic Heston model. We use the Adams method \citep{eleuch2019char, han2019mean} to solve the Riccati-Volterra equations (\ref{Eq:psi}) and (\ref{Eq:Exppsi}) numerically. Figures (\ref{Fig:PowA})-(\ref{Fig:ExpA}) show the $A_t$ in $\alpha^*$ (\ref{Eq:alpha*}) and $u^*$ (\ref{Eq:Expu}), under different values of $H$. Assumptions in Theorem \ref{Thm:alpha*} and \ref{Thm:Expu} are satisfied by the parameter setting detailed in figure descriptions. Figure (\ref{Fig:PowA}) exhibits that if the stock volatility is rougher, the investment demand (\ref{Eq:alpha*}) under power utility is smaller. It is partially due to the market leverage. In the equity market, the correlation $\rho$ is usually negative. $\psi(t)$ is positive for $t>0$. Moreover, the value of $\psi(t)$ becomes larger for a smaller $H$. Interestingly, exponential utility implies complete different opinion about roughness. It suggests investing more under smaller $H$.

We stress that Figures (\ref{Fig:PowA})-(\ref{Fig:ExpA}) are sensitivity analysis and demonstrate the marginal effect of roughness parameter. However, it is not plausible to only vary Hurst parameter $H$ while keeping other parameters unchanged. Therefore, we adopt a simulated implied volatility surface in \cite{abi2019lifting}. The investors calibrate two sets of parameters under classic Heston model and rough Heston model respectively, for the same implied volatility surface. The investor under the Heston model uses the calibrated parameters in \citet[Table 6]{abi2019lifting} to implement the optimal strategy. The investor under rough Heston model uses \citet[Table 4]{abi2019lifting} instead. Figures (\ref{Fig:PowIV})-(\ref{Fig:ExpIV}) illustrate the $A_t$ under power utility and exponential utility. Again, they have different opinions about investing less or more if rough Heston model is used.  Moreover, Figures (\ref{Fig:PowA})-(\ref{Fig:ExpA}) and Figures (\ref{Fig:PowIV})-(\ref{Fig:ExpIV}) do not conflict with each other. The calibrated parameters in \citet[Table 6]{abi2019lifting} for classic Heston model have significantly larger mean-reversion and volatility-of-volatility parameters than the counterparts in rough Heston model. 

There is no definite answer to the question that investors should invest more or less for different roughness, under expected utility theory. The type of risk aversions matters. Our discovery provides a new difference in utility preference by considering roughness. But the rationale behind is largely unexplained and needs more investigations. Moreover, expected utility theory also suffers some issues, see \citet[Chapter 5]{ait2009handbook} and references therein. 

\section{Concluding remarks}\label{Sec:Con}

In this paper, we solve Merton's portfolio optimization under Volterra Heston model. Economically,
this paper shows that different forms of risk aversion described by concave utility functions exhibit different
sensitivity in investment demand with respect to the volatility roughness. Technically, we offer a novel
solution approach for the Merton portfolio problem in a rough volatility economy based on the martingale
optimality principle. The key is the recognition of a novel auxiliary stochastic process. A future direction
is to incorporate model uncertainty.


\appendix
\section{Proofs of main results}\label{Appednix}
\begin{proof}[Proof of Theorem \ref{Thm:M}]
	First of all, we point out that there exists a unique continuous solution to (\ref{Eq:psi}) over $[0,T]$  under Assumption (\ref{Assum:M}). Then, we claim 
	\begin{equation}\label{Eq:Mequi}
	M^{1/c}_t = \tildeE_t \Big[ \exp \Big(\int^T_t \big(\frac{\gamma}{c} r_s + \frac{\gamma}{2 c (1 - \gamma)} \theta^2 V_s \big) ds \Big) \Big].
	\end{equation}
	Indeed, by \citet[Theorem 4.3]{abi2017affine},
	\begin{align*}
	&\exp \Big[ \int^T_t \big(\frac{\gamma \theta^2}{2c(1-\gamma)} \xi_t(s) + \frac{\sigma^2}{2} \psi^2(T-s) \xi_t(s) \big) ds \Big] = \tildeE_t \Big[ \exp \Big(\int^T_t  \frac{\gamma}{2 c (1 - \gamma)} \theta^2 V_s ds \Big)\Big].
	\end{align*}
	The martingale assumption in \citet[Theorem 4.3]{abi2017affine} is guaranteed by \citet[Lemma 7.3]{abi2017affine} for Volterra Heston model. 
	
	As $V_t$ is non-negative, $r_t > 0$ is deterministic, and $ 1-\gamma \leq c \leq 1$, we have $M_t \geq l > 0$ in view of (\ref{Eq:Mequi}).
	
	Let $ L \triangleq \exp \Big(\int^T_0  \frac{\gamma \theta^2}{2 c (1 - \gamma)}  V_s ds \Big)$ and the Radon-Nikodym derivative at $\cF_T$ as
	\begin{equation*}
	R = \exp\Big( - \frac{\gamma^2 \theta^2}{2(1-\gamma)^2} \int^T_0 V_t dt  + \frac{\gamma \theta}{1 - \gamma} \int^T_0 \sqrt{V_t} dW_{1t} \Big).
	\end{equation*}
	Then
	\begin{align*}
	\E\Big[ \sup_{ t \in [0, T]} |M_t|^p \Big] & \leq C\E\Big[  \sup_{ t \in [0, T]} \tildeE_t \Big[ \exp \Big(\int^T_t  \frac{\gamma \theta^2}{2 c (1 - \gamma)}  V_s ds \Big) \Big]^{pc} \Big] \\
	& \leq C\E\Big[  \sup_{ t \in [0, T]} \tildeE_t \Big[ \exp \Big(\int^T_0  \frac{\gamma \theta^2}{2 c (1 - \gamma)}  V_s ds \Big) \Big]^{pc} \Big] \leq C\tildeE \Big[  \sup_{ t \in [0, T]} \tildeE_t [L]^{2pc} \Big]^{1/2} \E\big[R^{-1}\big]^{1/2}.
	\end{align*}
	By \citet[Theorem 2.6 and Lemma A.2]{han2019mean}, we have $\tildeE[L] < \infty$ if $\lambda > 0$ and $\lambda^2 - \frac{ \gamma}{c(1 - \gamma)} \theta^2 \sigma^2 > 0$. Therefore, $\tildeE_t [L]$ is a martingale under $\tilde{\p}$. Note $2pc > 1$, by Doob's maximal inequality, 
	\begin{align*}
	\tildeE \Big[  \sup_{ t \in [0, T]} \tildeE_t [L]^{2pc} \Big] \leq C \tildeE \Big[ \exp \Big(\frac{p\gamma \theta^2}{1 - \gamma} \int^T_0 V_s ds \Big) \Big] < \infty.
	\end{align*}
	The last inequality holds under the assumption that $\lambda > 0$ and $\lambda^2 - 2 p \frac{ \gamma}{1 - \gamma} \theta^2 \sigma^2 > 0$. The argument is the same for $\tildeE[L] < \infty$.
	Moreover, by H\"older's inequality,
	\begin{align*}
	&\E\big[R^{-1}\big]  \leq \E \Big[ \exp \Big(\frac{3 \gamma^2 \theta^2}{(1 - \gamma)^2} \int^T_0 V_t dt \Big) \Big]^{1/2} \E \Big[ \exp\Big( - \frac{2 \gamma^2 \theta^2}{(1-\gamma)^2} \int^T_0 V_t dt  - \frac{2 \gamma \theta}{1 - \gamma} \int^T_0 \sqrt{V_t} dW_{1t} \Big) \Big]^{1/2} < \infty.
	\end{align*}
	$\E \big[ e^{\frac{3 \gamma^2 \theta^2}{(1 - \gamma)^2} \int^T_0 V_t dt} \big]$ is finite since $\kappa^2 - 6 \frac{\gamma^2}{(1 - \gamma)^2} \theta^2 \sigma^2 > 0$. Therefore, $\E\big[ \sup_{ t \in [0, T]} |M_t|^p \big] < \infty $ holds.
	
	For property (2), the proof is in the same spirit of \citet[Theorem 4.1 (2)]{han2019mean}. Let 
	\begin{equation}
	Z_t = \int^T_t \big[ \gamma r_s + \frac{\gamma \theta^2 \xi_t(s)}{2(1-\gamma)}  + \frac{c\sigma^2}{2} \psi^2(T-s) \xi_t(s) \big] ds.
	\end{equation}
	Then $M_t = e^{Z_t} $. Applying It\^o's lemma to $\xi_t(s)$ on time $t$ yields
	\begin{equation}
	d \xi_t(s) = \frac{1}{\lambda} R_{\lambda}(s-t) \sigma \sqrt{V_t} d \tilde B_t
	\end{equation}
	from (\ref{Eq:xi}).
	Then
	\begin{align*}
	d Z_t  = &\big[ - \gamma r_t - \frac{\gamma \theta^2}{2(1-\gamma)} V_t - \frac{c\sigma^2}{2} \psi^2(T-t) V_t \big] dt \\
	&  + \frac{\gamma \theta^2}{2(1-\gamma)} \int^T_t \frac{1}{\lambda} R_{\lambda}(s-t) \sigma \sqrt{V_t} d \tilde B_t ds + \frac{c \sigma^2}{2} \int^T_t  \psi^2(T-s) \frac{1}{\lambda} R_{\lambda}(s-t) \sigma \sqrt{V_t} d \tilde B_t ds \\
	= &\big[ - \gamma r_t - \frac{\gamma \theta^2}{2(1-\gamma)} V_t - \frac{c\sigma^2}{2} \psi^2(T-t) V_t \big] dt \\
	& +  \frac{\gamma \theta^2}{2(1-\gamma)} \int^T_t \sigma \frac{1}{\lambda} R_{\lambda}(s-t) ds \sqrt{V_t} d \tilde B_t + \frac{c \sigma^2}{2} \int^T_t  \sigma \psi^2(T-s) \frac{1}{\lambda} R_{\lambda}(s-t) ds \sqrt{V_t} d \tilde B_t.
	\end{align*}
	The second equality relies on the stochastic Fubini theorem \citep{veraar2012fubini}.
	
	Next, we show
	\begin{equation}\label{Eq:Uequivalent}
	\int^T_t \Big[ \frac{c \sigma^2}{2} \psi^2(T-s) + \frac{\gamma \theta^2}{2(1-\gamma)}  \Big] \frac{1}{\lambda} R_{\lambda}(s-t) ds = c \psi (T-t). 
	\end{equation}
	
	In fact,
	\begin{align*}
	& \int^T_t \big[ \frac{c \sigma^2}{2} \psi^2(T-s) + \frac{\gamma \theta^2}{2(1-\gamma)} \big] \frac{1}{\lambda} R_{\lambda}(s-t) ds - c \psi (T-t) \\
	& = \big[ \frac{c \sigma^2}{2} \psi^2 + \frac{\gamma \theta^2}{2(1-\gamma)} \big] * \frac{1}{\lambda} R_{\lambda}(T - t) - c K* \big[\frac{\sigma^2}{2} \psi^2 - \lambda \psi+ \frac{\gamma \theta^2}{2 c (1-\gamma)} \big](T-t) \\
	& = \big[ \frac{c\sigma^2}{2} \psi^2 + \frac{\gamma \theta^2}{2(1-\gamma)} \big] * \big[ \frac{1}{\lambda} R_{\lambda} - K \big](T - t) + c \lambda K*\psi(T-t)\\
	& = - R_\lambda*K*\big[ \frac{c\sigma^2}{2} \psi^2 + \frac{\gamma \theta^2}{2(1-\gamma)} \big](T-t) + c \lambda K*\psi(T-t) \\
	& =  c \big[ \lambda K - R_\lambda - \lambda K* R_\lambda \big]* \psi(T-t) = 0.
	\end{align*}
	We have used the equality
	\begin{equation}
	R_\lambda * \psi = R_\lambda*K*\big[ \frac{\sigma^2}{2} \psi^2 - \lambda \psi + \frac{\gamma \theta^2}{2c(1-\gamma)} \big].
	\end{equation}
	
	Therefore, 
	\begin{align*}
	d M_t = & M_t d Z_t + \frac{1}{2} M_t dZ_t dZ_t \\
	= & M_t \big[ - \gamma r_t - \frac{\gamma \theta^2}{2(1-\gamma)} V_t - \frac{c\sigma^2}{2} \psi^2(T-t) V_t \big] dt + \frac{U^2_{1t} + U^2_{2t}}{2 M_t} dt + U_{1t} d\tilde W_{1t} + U_{2t} d W_{2t} \\
	=& - \big[ \gamma r_t + \frac{\gamma}{2(1-\gamma)}\theta^2 V_t \big] M_t dt - \frac{\gamma \theta}{1 - \gamma} \sqrt{V_t} U_{1t} dt - \frac{\gamma}{2(1-\gamma)}\frac{U^2_{1t}}{M_t} dt + U_{1t} dW_{1t} + U_{2t} dW_{2t}.
	\end{align*}
	
	Finally, for property (3),
	\begin{align*}
	& \E\Big[ \big(\int^T_0 U^2_{it} dt \big)^{p/4} \Big] \leq C \E\Big[ \sup_{ t \in [0, T]} |M_t|^{p/2} \big(\int^T_0 V_t dt \big)^{p/4} \Big] \leq C \E\Big[ \sup_{ t \in [0, T]} |M_t|^p \Big]^{1/2} \E\Big[ e^{a \int^T_0 V_t dt} \Big]^{1/2} < \infty,
	\end{align*}
	where $ a > 0 $ is constant.
\end{proof}

\begin{proof}[Proof of Theorem \ref{Thm:alpha*}]
	(1). Clearly, as  $M_T = 1$, $J^\alpha_T = \frac{X^\gamma_T}{\gamma}$.
	
	(2). As $M_0$ is a constant independent of $ \alpha \in \cA$, $J^\alpha_0 = \frac{x^\gamma_0}{\gamma} M_0$ is a constant independent of $ \alpha \in \cA$.
	
	(3). By It\^o's lemma,
	\begin{align*}
	dJ^\alpha_t = & \big[\theta \alpha_t \sqrt{V_t} + \frac{\gamma - 1}{2} \alpha^2_t \big] M_t X^\gamma_t dt + \alpha_t X^\gamma_t U_{1t} dt - \frac{1}{2(1-\gamma)} \big[ \theta^2 V_t M_t + 2 \theta \sqrt{V_t} U_{1t} + \frac{U^2_{1t}}{M_t} \big] X^\gamma_t dt  \\
	& + \big[ \frac{X^\gamma_t U_{1t}}{\gamma} + M_t \alpha_t X^\gamma_t \big] dW_{1t} + \frac{X^\gamma_t U_{2t}}{\gamma} dW_{2t} \\
	\triangleq & J^\alpha_t  F(\alpha, t) dt + J^\alpha_t \big[ \frac{U_{1t}}{M_t} + \gamma \alpha_t \big] dW_{1t} + J^\alpha_t \frac{U_{2t}}{M_t} dW_{2t},
	\end{align*}
	where 
	\begin{align*}
	F(\alpha, t) = & \frac{\gamma(\gamma - 1)}{2} \alpha^2 + \gamma \big[\theta \sqrt{V_t} + \frac{U_{1t}}{M_t} \big] \alpha - \frac{\gamma}{2(1-\gamma)} \big[ \theta \sqrt{V_t} +  \frac{U_{1t}}{M_t} \big]^2.
	\end{align*}
	$\alpha^*$ in (\ref{Eq:alpha*}) is derived from $\frac{\partial F}{\partial \alpha} = 0$. Note $F(\alpha, t)$ is a quadratic function on $\alpha$ and $\gamma - 1 < 0$. Since $F(\alpha^*, t) = 0$, then $F(\alpha, t) \leq 0$.
	
	Moreover, $J^\alpha_t = \frac{M_0 x^\gamma_0}{\gamma} e^{\int^t_0 F(\alpha_s, s) ds} G_t$, where
	\begin{equation*}
	G_t  = \exp \Big\{ - \frac{1}{2} \int^t_0 [ (\frac{U_{1s}}{M_s} + \gamma \alpha_s)^2 + \frac{U^2_{2s}}{M^2_s} ]ds + \int^t_0 [\frac{U_{1s}}{M_s} + \gamma \alpha_s] dW_{1s} + \int^t_0 \frac{U_{2s}}{M_s} dW_{2s} \Big\}.
	\end{equation*}
	$e^{\int^t_0 F(\alpha_s, s) ds}$ is non-increasing. Since $\int^t_0 \alpha^2_s ds < \infty$, $\p$-$\as$, the stochastic exponential $G_t$ is a local martingale. There exists a sequence of stopping times $\{\tau_n\}_{n=1,2,3...}$ satisfying $\lim_{n \rightarrow \infty} \tau_n = T$, $\p$-$\as$, such that 
	\begin{equation*}
	\E[J^\alpha_{t\wedge \tau_n} | \cF_s] \leq J^\alpha_{s\wedge\tau_n}, \quad s \leq t,
	\end{equation*}
	for every $n$.  Moreover, $J^\alpha_t$ is bounded below by 0. Let $n\rightarrow \infty$ and by Fatou's lemma, we deduce $J^\alpha_t$ is a supermartingale.
	
	For $\alpha_t = \alpha^*_t$, $G_t$ is a martingale by \citet[Lemma 7.3]{abi2017affine}. Subsequently, $J^{\alpha^*}_t$ is a true martingale. We have verified all conditions required by martingale optimality principle, except for the admissibility of $\alpha^*$.
	
	By Doob's maximal inequality,
	\begin{align*}
	\E \Big[ \sup_{ t \in [0, T]} | X^*_t |^q \Big] & \leq C \E \Big[ \sup_{ t \in [0, T]} e^{2q \int^t_0 \theta A_s V_s ds} \Big]^{1/2} \E \Big[ \sup_{ t \in [0, T]} \Big|  \exp \Big(- \int^t_0 \frac{A^2_s}{2}  V_s ds + \int^t_0 A_s \sqrt{V_s} dW_{1s} \Big) \Big|^{2q}\Big]^{1/2} \\
	& \leq C \E \Big[ e^{2q \int^T_0 |\theta A_s| V_s ds} \Big]^{1/2} \E \Big[ \exp \Big(- q\int^T_0 A^2_s  V_s ds + 2q \int^T_0 A_s \sqrt{V_s} dW_{1s} \Big) \Big]^{1/2}.
	\end{align*}
	The first term is finite. The second term is also finite. In fact, by H\"older's inequality,
	\begin{align*}
	& \E \Big[ \exp \Big(- \int^T_0 q A^2_s  V_s ds + \int^T_0 2q A_s \sqrt{V_s} dW_{1s} \Big) \Big] \\
	& \leq  \E \Big[ e^{(8q^2 - 2q) \int^T_0 A^2_s  V_s ds} \Big]^{1/2} \E\Big[ \exp \Big(- 8 q^2 \int^T_0 A^2_s  V_s ds + 4q \int^T_0 A_s \sqrt{V_s} dW_{1s} \Big) \Big]^{1/2} < \infty.
	\end{align*}
	$\E \Big[ \sup_{ t \in [0, T]} | X^*_t |^q \Big] < \infty$ is proved. It becomes straightforward to verify $\alpha^*$ is admissible. 
\end{proof}

\begin{proof}[Proof of Theorem \ref{Thm:ExpM}]
		It is straightforward to see that $M_t > 0$ in (\ref{Eq:ExpM}). As for the upper bound, if $1 - \rho^2 = 0$, then $M_t \leq 1$, $\p$-$\as$. If $1 - \rho^2 > 0$, we have
	\begin{equation}\label{Eq:TransM}
	M^{1-\rho^2}_t = \tildeE \Big[ e^{- \frac{\theta^2(1 - \rho^2)}{2} \int^T_t V_s ds} \Big| \cF_t \Big],
	\end{equation}
	in the same spirit of \citet[Theorem 4.1]{han2019mean}. Then Property (1) is proved.
	
	For Property (2), denote $ M_t = e^{Z_t} $ in (\ref{Eq:ExpM}) with proper $Z_t$. Then
	\begin{align*}
	d Z_t  = &\big[ \frac{\theta^2}{2} - \frac{(1 - \rho^2) \sigma^2}{2} \psi^2(T-t) \big] V_t dt + d \tilde B_t \cdot \sigma \sqrt{V_t}   \int^T_t \Big[ \frac{(1-\rho^2) \sigma^2}{2} \psi^2(T-s) - \frac{\theta^2}{2}  \Big] \frac{1}{\lambda} R_{\lambda}(s-t) ds \\
	= &\big[ \frac{\theta^2}{2} - \frac{(1 - \rho^2) \sigma^2}{2} \psi^2(T-t) \big] V_t dt +\psi(T - t) \sigma \sqrt{V_t} d \tilde B_t .
	\end{align*}
	Applying It\^o's lemma to $M_t= e^{Z_t}$ with function $f(z) = e^z$ gives (\ref{Eq:ItoExpM}).
	
	Proof of Property (3) is exactly the same as \citet[Theorem 4.1]{han2019mean}.
\end{proof}

\begin{proof}[Proof of Theorem \ref{Thm:Expu}]
The first two conditions are straightforward. For (3), by It\^o's lemma,
\begin{align*}
dJ^u_t \triangleq & J^u_t  F(u_t, t) dt + J^u_t \big[ \frac{U_{1t}}{M_t} -\gamma e^{\int^T_t r_v dv} u_t \big] dW_{1t} + J^u_t \frac{U_{2t}}{M_t} dW_{2t},
\end{align*}
where 
\begin{align*}
F(u, t) = & \frac{\gamma^2}{2} e^{\int^T_t 2r_v dv} u^2 - \gamma e^{\int^T_t r_v dv} \big[\theta \sqrt{V_t} + \frac{U_{1t}}{M_t} \big] u + \frac{1}{2} \big[ \theta \sqrt{V_t} +  \frac{U_{1t}}{M_t} \big]^2.
\end{align*}
$u^*$ in (\ref{Eq:Expu}) is derived from $\frac{\partial F}{\partial u} = 0$. Note $F(u^*, t) = 0$, then $F(u, t) \geq 0$.

Moreover, $J^u_t = - \frac{M_0}{\gamma} \exp\big[ - \gamma e^{ \int^T_0 r_v dv} x_0 \big] e^{\int^t_0 F(u_s, s) ds} G_t$,
\begin{equation*}
 G_t  = \exp \Big\{ - \frac{1}{2} \int^t_0 [ (\frac{U_{1s}}{M_s} -\gamma e^{\int^T_s r_v dv} u_s)^2 + \frac{U^2_{2s}}{M^2_s} ]ds + \int^t_0 [\frac{U_{1s}}{M_s} -\gamma e^{\int^T_s r_v dv} u_s] dW_{1s} + \int^t_0 \frac{U_{2s}}{M_s} dW_{2s} \Big\}.
\end{equation*}
Since $u(\cdot)$ is admissible, $G_t$ is a local martingale. There exists a sequence of stopping times $\{\tau_n\}_{n = 1,2,3,...}$ and $\lim_{n \rightarrow \infty} \tau_n = T$, $\p$-$\as$, such that $G_{t \wedge \tau_n}$ is a positive martingale for every $n$. Furthermore, $- \frac{M_0}{\gamma} \exp\big[ - \gamma e^{ \int^T_0 r_v dv} x_0 \big] e^{\int^t_0 F(u_s, s) ds}$ is non-increasing. Therefore, $J^u_{t\wedge \tau_n}$ is a supermartingale. Then for $s \leq t$, $\E[ J^u_{t\wedge \tau_n} | \cF_s] \leq J^u_{s\wedge \tau_n}$. It implies that for any set $ A \in \cF_s$, 
\begin{equation*}
\E[ J^u_{t\wedge \tau_n} \id_A] \leq \E[J^u_{s\wedge \tau_n}\id_A].
\end{equation*}
Since $\{\exp\big[ - \gamma e^{ \int^T_{t\wedge \tau_n} r_v dv} X_{t \wedge \tau_n} \big]\}_n$ is uniformly integrable and $M$ is bounded, $\{J^u_{t\wedge \tau_n} \}_n$ and $\{J^u_{s\wedge \tau_n} \}_n$ are uniformly integrable. Let $n \rightarrow \infty$, then $\E[ J^u_t \id_A] \leq \E[J^u_s \id_A]$. Then we deduce that $J^u$ is a supermartingale.

For $u_t = u^*_t$, $G_t$ is a martingale by \citet[Lemma 7.3]{abi2017affine}. Subsequently, $J^{u^*}_t$ is a true martingale. For the admissibility of $u^*$, first,
\begin{align*}
\E\Big[\int^T_0 u^2_t dt \Big] \leq C \E\Big[ \int^T_0 V_t dt \Big] \leq C\sup_{t \in [0, T]} \E\big[  V_t \big] < \infty.
\end{align*}	
The last term is finite by \citet[Lemma 3.1]{abi2017affine}.

To prove that $\{\exp\big[ - \gamma e^{ \int^T_\tau r_v dv} X^*_\tau \big]: \tau \text{ stopping time with values in } [0, T] \}$ is a uniformly integrable family, we only need to show
\begin{equation}
\sup_{\tau}\E\Big[ e^{- p\gamma e^{ \int^T_\tau r_v dv} X^*_\tau} \Big] < \infty, \text{ for some } p > 1.
\end{equation}
Note
\begin{align*}
X^*_\tau =& e^{\int^\tau_0 r_sds}x_0 + \int^\tau_0 e^{ \int^\tau_s r_v dv} \theta A_s V_s ds + \int^\tau_0 e^{\int^\tau_s r_v dv} A_s \sqrt{V_s} dW_{1s},
\end{align*}
then
\begin{align*}
& \sup_{\tau}\E\Big[ e^{- p\gamma e^{ \int^T_\tau r_v dv} X^*_\tau} \Big] \leq C \sup_{\tau} \E \Big[ \exp \Big(- p \gamma \int^\tau_0 e^{ \int^T_s r_v dv} \theta A_s V_s ds - p \gamma \int^\tau_0 e^{\int^T_s r_v dv} A_s \sqrt{V_s} dW_{1s} \Big) \Big] \\
& \leq C \sup_{\tau} \E \Big[ \exp \Big(p \gamma \int^T_0 e^{ \int^T_s r_v dv} |\theta A_s| V_s ds + p^2\gamma^2 \int^T_0 e^{\int^T_s 2r_v dv} A^2_s V_s ds \\
&\hspace{3cm} - p^2\gamma^2 \int^\tau_0 e^{\int^T_s 2r_v dv} A^2_s V_s ds - p \gamma \int^\tau_0 e^{\int^T_s r_v dv} A_s \sqrt{V_s} dW_{1s} \Big) \Big] \\
& \leq  C \E \Big[ \exp\Big(\int^T_0 (2 p^2 \gamma^2 e^{ \int^T_s 2 r_v dv} A^2_s + 2 p \gamma e^{ \int^T_s r_v dv} |\theta A_s|) V_s ds \Big) \Big]^{1/2} \\
& \quad \times \sup_{\tau} \E\Big[ \exp \Big(- 2 p^2 \gamma^2 \int^\tau_0 e^{\int^T_s 2 r_v dv} A^2_s V_s ds  -2 p \gamma \int^\tau_0 e^{\int^T_s r_v dv} A_s \sqrt{V_s} dW_{1s} \Big) \Big]^{1/2} \\
& < \infty.
\end{align*}
The last inequality holds by \citet[Lemma A.2 and Theorem 2.6]{han2019mean} with assumption $\kappa^2 - 2 \eta \sigma^2 > 0$ and optional sampling theorem with the fact that $\tau \leq T$.
\end{proof}

\end{document}